\documentclass[11pt,letterpaper]{article}
\usepackage[T1]{fontenc}
\usepackage{lmodern}
\usepackage[margin=1in]{geometry}
\usepackage{amsmath,amssymb,amsthm,mathtools}
\usepackage{microtype,booktabs,array}
\usepackage{enumitem,fancyhdr,needspace}
\usepackage[hidelinks]{hyperref}
\hypersetup{pdftitle={Euclidean SVP is deterministically NP-hard to approximate within any constant factor},pdfauthor={Daqing Wan},pdfsubject={Revised proof by rank and support tensor amplification}}
\numberwithin{equation}{section}
\newtheorem{theorem}{Theorem}[section]
\newtheorem{lemma}[theorem]{Lemma}
\newtheorem{proposition}[theorem]{Proposition}
\theoremstyle{definition}
\newtheorem{problem}[theorem]{Problem}
\theoremstyle{remark}
\newtheorem{remark}[theorem]{Remark}
\newtheorem*{remark*}{Remark}
\DeclareMathOperator{\supp}{supp}
\DeclareMathOperator{\rank}{rank}
\DeclareMathOperator{\tr}{tr}
\DeclareMathOperator{\Tr}{Tr}
\DeclareMathOperator{\diag}{diag}
\newcommand{\Z}{\mathbb Z}
\newcommand{\R}{\mathbb R}
\newcommand{\F}{\mathbb F}
\newcommand{\E}{\mathbb E}
\newcommand{\cB}{\mathcal B}
\newcommand{\YES}{\mathrm{YES}}
\newcommand{\NO}{\mathrm{NO}}
\setlist[itemize]{topsep=4pt,itemsep=2pt}
\allowdisplaybreaks[1]
\title{\textbf{Euclidean SVP is deterministically NP-hard\\
to approximate within any constant factor}}
\author{Daqing Wan\thanks{ChatGPT assisted with drafting and proof revision.}}
\date{}
\begin{document}
\maketitle
\begin{abstract}
Using the deterministic binary Reed--Solomon projection theorem from our earlier work, we prove that, for every constant $\rho>1$, Euclidean $\rho$-GapSVP is NP-hard under a deterministic polynomial-time many-one reduction. The same conclusion holds in every fixed finite $\ell_p$ norm, with the usual effectiveness convention for a fixed real $p$. The reduction from Gap Exact Set Cover produces an integer lattice whose NO instances have $\ell_1$ minimum at least $2k$ and Hamming distance at least $k+1$ after reduction modulo a prime $q$. A two-factor Euclidean inequality supplies the initial tensor bound. We then combine an integer collision inequality, modular Hamming distance, and a low-rank determinant estimate to amplify the squared Euclidean minimum. Repeating this argument with larger tensor blocks gives any prescribed constant gap. The freedom to choose the gadget exponent after the tensor order keeps the required prime polynomially bounded. An explicit sign embedding proves the extension to all fixed finite norms. The reduction, amplification, and norm-transfer proofs are included in full; the arithmetic projection theorem is stated precisely as the external input.
\end{abstract}

\section{Introduction}
Let $L\subseteq\R^N$ be a lattice. Its Euclidean minimum is
\[
 \lambda_2(L)=\min_{0\ne x\in L}\|x\|_2.
\]
The shortest vector problem asks for a nonzero lattice vector attaining this minimum. We study its approximation promise problem.

\begin{problem}[Euclidean $\rho$-GapSVP]\label{prob:svp}
Fix $\rho\ge1$. Given an integer lattice basis and a positive rational number $d$, distinguish
\[
 \YES:\ \lambda_2(L)\le d,
 \qquad
 \NO:\ \lambda_2(L)>\rho d,
\]
under the promise that one of these cases holds.
\end{problem}

Ajtai established exact Euclidean SVP hardness under randomized reductions~\cite{ajtai}, Micciancio obtained hardness for every fixed factor below $\sqrt2$~\cite{mic01}, and Khot extended randomized hardness to arbitrary constants~\cite{khot}. Haviv and Regev used lattice tensor products to obtain stronger, dimension-dependent randomized hardness~\cite{hr}. Their analysis controls general tensor vectors through properties of the constructed lattices: a shortest vector in a tensor product need not be a pure tensor.

Our earlier work~\cite{wan} constructs a deterministic Reed--Solomon locally dense lattice with a surjective binary projection. Every prescribed binary pattern occurs as the projection of a binary vector of one fixed weight in a specified lattice coset. We use this result in the parameter form stated in Theorem~\ref{thm:gadget}. Together with the classical hardness of Gap Exact Set Cover, it gives the following theorem.

\begin{theorem}[Euclidean hardness]\label{thm:main}
For every constant $\rho>1$, Euclidean $\rho$-GapSVP is NP-hard under a deterministic polynomial-time many-one reduction.
\end{theorem}

\paragraph{Changes from the first version.}
The tensor amplification part of the first version relies on the Lee-metric tensor-product formula in the published paper of N.~Ojiro and H.~Matsui~\cite{om} to claim exact multiplication of the $\ell_1$ minimum of integer lattices. That Lee-metric formula is false, and the resulting lattice identity is false as well. Appendix~\ref{app:counterexample} gives an explicit counterexample to both statements, including the quotient-code formulation with all the relevant modulus conditions. This version replaces that argument with a different tensor amplification proof based on Euclidean length, matrix rank, and Hamming distance modulo a prime. We retain the block reduction and the Reed--Solomon gadget, and include their statements and the necessary soundness proofs here. The results proved in this version concern every fixed approximation factor; the dimension-dependent claims of the first version are not retained.

\subsection{The reduction and its two soundness bounds}
For an integer-coordinate lattice, write
\[
 \lambda(L)=\min_{0\ne x\in L}\|x\|_1.
\]
For every fixed rational $0<\varepsilon<1$, the reduction constructs a lattice $L\subseteq\Z^N$, a prime $q$, and an integer $k$ such that
\begin{equation}\label{eq:intro-gap}
\begin{aligned}
 \YES:&\quad \text{some }0\ne y\in L\cap\{0,1\}^N
                   \text{ has }\|y\|_1<(1+\varepsilon)k,\\
 \NO:&\quad \lambda(L)\ge2k,\qquad d_H(L\bmod q)\ge k+1.
\end{aligned}
\end{equation}
Here $L\bmod q$ is a linear code over $\F_q$, and $d_H$ denotes minimum nonzero Hamming weight. The integer lower bound comes from the block reduction. Its support argument also works over $\F_q$, and this gives the modular bound needed for amplification.

Minimum Hamming distance multiplies under tensor products of linear codes over a field. This is the coding-theoretic amplification mechanism used, for example, by Cheng and Wan~\cite{cw}. Our proof uses that field-code property and estimates integer tensor vectors separately.

\subsection{The amplification argument}
For integer-coordinate lattices $U,V$, a row-and-column argument proves
\begin{equation}\label{eq:intro-seed}
 \lambda_2(U\otimes V)^2\ge\lambda(U)\lambda(V).
\end{equation}
Thus a NO lattice in \eqref{eq:intro-gap} satisfies
\[
 \lambda_2(L^{\otimes2})^2\ge4k^2,
 \qquad d_H(L^{\otimes2}\bmod q)\ge(k+1)^2.
\]
This tensor square is the starting point of the proof.

More generally, suppose an integer $a\ge4$ and a distance scale $D$ satisfy
\[
 \lambda_2(U)^2\ge aD,\qquad d_H(U\bmod q)\ge D.
\]
Theorem~\ref{thm:phase} shows that, for any prescribed finite tensor order and a sufficiently large prime,
\[
 \lambda_2(U^{\otimes j})^2\ge c_jD^j,
 \qquad c_j\longrightarrow\frac{a(a+1)}3
 \quad\text{increasingly}.
\]
The proof views a hypothetical short tensor as an integer matrix. If its rank modulo $q$ is large, the Hamming bound supplies many short columns whose reductions are independent. Such columns differ in many coordinates, and an integer collision inequality then forces too much squared length. If its modular rank is small, a bound on integer minors transfers this rank bound to the reals. A Euclidean determinant estimate rules out the remaining low-rank case.

One application has a finite limiting coefficient. To go further, we use the amplified tensor block as the input to the same theorem with a larger integer coefficient. The coefficients can successively reach
\[
 4,\ 6,\ 12,\ 24,\ 48,\ldots.
\]
For any prescribed constant bound, Proposition~\ref{prop:schedule} therefore gives an integer $b$ above that bound and an even tensor order $T$ such that
\begin{equation}\label{eq:intro-tensor}
 \lambda_2(L^{\otimes T})^2\ge bk^T
\end{equation}
in NO instances. In YES instances, the binary vector $y^{\otimes T}$ has squared norm less than $(1+\varepsilon)^Tk^T$. We first choose $b>\rho^2$ and its tensor schedule, then choose $\varepsilon$ so that $\rho^2(1+\varepsilon)^T<b$. The gadget exponent is chosen after these constants. This order makes all prime requirements compatible with a polynomial-time construction.

\begin{theorem}[All fixed finite norms]\label{thm:allp}
For every fixed finite effective $p\ge1$ and every constant $\rho>1$, $\rho$-GapSVP$_p$ is NP-hard under a deterministic polynomial-time many-one reduction.
\end{theorem}

The fixed-$p$ convention is given in Section~\ref{sec:prelim}; it includes every fixed rational $p\ge1$. For $p\ge2$, the binary YES vector and integrality give a direct proof. An explicit deterministic sign embedding also proves a norm transfer for all fixed finite $p$, including $p=1$.

\paragraph{Related deterministic hardness results.}
For every fixed $p>2$ and $\varepsilon>0$, Hair and Sahai~\cite{hs-high} proved NP-hardness of approximating SVP$_p$ within $2^{(\log n)^{1-\varepsilon}}$ under deterministic polynomial-time Karp reductions. Hecht and Safra~\cite{hecht-safra} independently obtained deterministic hardness for SVP$_p$ and unique-SVP$_p$: their constant-factor results apply to every $p>2$ under assumptions excluding specified deterministic subexponential-time algorithms, and their stronger approximation gaps apply to sufficiently large $p$. More recently, Hair and Sahai~\cite{hs-poly} obtained deterministic polynomial-time NP-hardness within $n^\varepsilon$ for every fixed $2<p<\infty$ and every $0<\varepsilon<\min\{(p-2)/(4p),1/8\}$, where $n$ is the lattice rank. They also obtain every exponent $\varepsilon<1/8$ for $p=\infty$; earlier deterministic hardness in that norm was established by Dinur~\cite{dinur}.

For all fixed finite $\ell_p$ norms, Hair and Sahai~\cite{hs-all} proved hardness within $2^{(\log n)^{1-o(1)}}$ by deterministic subexponential-time reductions, assuming $\mathrm{NP}\nsubseteq\bigcap_{\delta>0}\mathrm{DTIME}(\exp(n^\delta))$. Thus deterministic hardness for all finite norms was already known under that stronger complexity assumption. Theorem~\ref{thm:allp} gives deterministic polynomial-time NP-hardness for every fixed approximation factor in every fixed finite norm.

\subsection{Organization and external input}
Section~\ref{sec:prelim} fixes conventions. Section~\ref{sec:base} gives the source reduction, states the gadget, and proves both bounds in \eqref{eq:intro-gap}. Sections~\ref{sec:euclidean} and~\ref{sec:amplification} establish the tensor estimates. Section~\ref{sec:hardness} chooses the parameters and completes the Euclidean reduction. Section~\ref{sec:norms} proves the finite-norm extension, and Section~\ref{sec:scope} explains the scope of the results.

The only imported nonclassical lattice result is the binary projection theorem from~\cite[Theorem 7.3 and its preceding proof]{wan}. Its arithmetic-geometric proof is not reproduced here. We also use the standard NP-hardness of Gap Exact Set Cover. The lattice reduction, tensor estimates, and norm transfer needed after these inputs are proved below.

\section{Preliminaries}\label{sec:prelim}
\subsection{Lattices, norms, and encoding}
For a full-column-rank matrix $B\in\Z^{N\times n}$, let
\[
 L(B)=B\Z^n\subseteq\Z^N.
\]
Its rank is $n$ and its ambient dimension is $N$. In this paper, an integer lattice has integer coordinates in the specified embedding. Put
\[
 \lambda_p(L)=\min_{0\ne x\in L}\|x\|_p\quad(1\le p<\infty),
 \qquad \lambda(L)=\lambda_1(L).
\]
The problem $\rho$-GapSVP$_p$ is Problem~\ref{prob:svp} with $\lambda_p$ in place of $\lambda_2$.

For a fixed real $p$, we require that upper and lower rational approximations to $u^{1/p}$, to any fixed relative accuracy, can be computed in time polynomial in the bit length of a positive rational $u$. This holds for fixed rational $p$ and for fixed polynomial-time computable real $p\ge1$. Every use of such an approximation will have fixed strict slack. The Euclidean output radius is rational exactly.

For integer and binary vectors, respectively,
\begin{equation}\label{eq:integer-norms}
 \|z\|_2^2\ge\|z\|_1\quad(z\in\Z^N),
 \qquad \|y\|_p^p=\|y\|_1\quad(y\in\{0,1\}^N).
\end{equation}
For $p\ge2$, also $\|z\|_p^p\ge\|z\|_2^2$. These facts follow coordinate by coordinate.

\subsection{Tensor products and matrix rank}
For column bases $B\in\Z^{N\times n}$ and $C\in\Z^{M\times m}$, define
\[
 L(B)\otimes L(C)=L(B\otimes C)\subseteq\Z^{NM}.
\]
A change of lattice basis uses a unimodular integer matrix, so this definition is independent of the chosen bases. With a consistent ordering of the tensor coordinates, every vector can be reshaped as
\begin{equation}\label{eq:matrix-form}
 W=BZC^T,\qquad Z\in\Z^{n\times m}.
\end{equation}
The matrix $W$ has $N$ rows and $M$ columns. In the factorization $W=B(ZC^T)$, the matrix $ZC^T$ has integer entries. Each column of $W$ is therefore an integer linear combination of the columns of $B$, and belongs to $L(B)$. Similarly,
\[
 W^T=C(Z^TB^T),
\]
where $Z^TB^T$ has integer entries. Thus every column of $W^T$, equivalently every transposed row of $W$, belongs to $L(C)$. The squared Euclidean norm of the original tensor vector is the squared Frobenius norm of its matrix form:
\[
 \|W\|_F^2=\sum_{i=1}^N\sum_{j=1}^M W_{ij}^2.
\]
Here $i$ indexes the rows and $j$ indexes the columns, so the sum includes all entries of $W$.
For a pure tensor, $\|x\otimes y\|_p=\|x\|_p\|y\|_p$. A general lattice tensor is an integer sum of pure tensors, so the NO analysis must handle all matrices in \eqref{eq:matrix-form}.

The real rank of an integer matrix is its rank over $\R$, equivalently over $\mathbb Q$, or the largest order of a minor with nonzero integer determinant. Reduction modulo a prime may make such a determinant vanish, and hence
\[
 \rank_{\F_q}(W\bmod q)\le\rank_{\R}W.
\]
We will obtain a converse up to a fixed cutoff by showing that the relevant nonzero minors are too small to be divisible by $q$.

If $B$ has rank $n$, its $t$-fold Kronecker power has rank $n^t$. If $B$ has $N$ rows and entries of at most $s$ bits, this power has $N^t$ rows and entries of at most $ts+O(t)$ bits. A fixed tensor order therefore preserves polynomial encoding length.

\subsection{Reduction modulo a prime}
For $L\subseteq\Z^N$, write
\[
 L\bmod q=\{x\bmod q:x\in L\}\subseteq\F_q^N.
\]
This is the linear code generated by the reductions of the basis columns. For a linear code $C$, let $d_H(C)$ be its minimum nonzero support size, with $d_H(\{0\})=\infty$. Reducing tensor generators gives
\begin{equation}\label{eq:mod-tensor}
 (U\otimes V)\bmod q=(U\bmod q)\otimes_{\F_q}(V\bmod q).
\end{equation}
This statement concerns the images of generators and does not require reduction to be injective.

\begin{lemma}[Hamming distance under tensoring]\label{lem:hamming}
If linear codes $C_1,C_2$ over a field have minimum distances at least $D_1,D_2$, then every nonzero word of $C_1\otimes C_2$ has weight at least $D_1D_2$. Consequently,
\[
 d_H(U\bmod q)\ge D
 \quad\Longrightarrow\quad
 d_H(U^{\otimes j}\bmod q)\ge D^j.
\]
\end{lemma}
\begin{proof}
Reshape a nonzero tensor word as a matrix with columns in $C_1$ and rows in $C_2$. A nonzero column has at least $D_1$ nonzero coordinates, so at least $D_1$ rows are nonzero. Each such row has at least $D_2$ nonzero coordinates. Summing gives the product lower bound. Iteration and \eqref{eq:mod-tensor} give the last assertion. For nonzero codes, a pure tensor of minimum-weight words attains the product, so their distances in fact multiply exactly.
\end{proof}

\section{A deterministic binary gap with modular soundness}\label{sec:base}
\subsection{Gap Exact Set Cover}
A Set Cover instance consists of a nonempty universe, which we label $\{1,\ldots,u\}$, and subsets $S_1,\ldots,S_r$ of this universe. Its incidence matrix $S\in\{0,1\}^{u\times r}$ has $S_{ij}=1$ if $i\in S_j$, and $S_{ij}=0$ otherwise, for $1\le i\le u$ and $1\le j\le r$. Thus rows correspond to universe elements and columns correspond to the given subsets. For $c\in\{0,1\}^r$, the equation $Sc=\mathbf1$ means that the selected sets cover every universe element exactly once.

\begin{problem}[$\Gamma$-Gap Exact Set Cover]\label{prob:cover}
Fix $\Gamma>1$. Given $S$ and an integer $1\le\tau\le r$, distinguish:
\begin{itemize}
\item YES: some $c\in\{0,1\}^r$ has $Sc=\mathbf1$ and $\|c\|_1\le\tau$;
\item NO: every ordinary cover uses more than $\Gamma\tau$ sets.
\end{itemize}
\end{problem}
For every fixed $\Gamma>1$, this promise problem is NP-hard under deterministic polynomial-time reductions~\cite[Definition 2.5 and Theorem 2.6]{mic12}. We take $\Gamma$ to be an integer.

\begin{lemma}[Support soundness over the integers and a field]\label{lem:cover}
Form
\begin{equation}\label{eq:cover-matrix}
 B_0=\begin{pmatrix}I_r\\ S\\ \vdots\\ S\end{pmatrix},
 \qquad t_0=\begin{pmatrix}\mathbf0\\ \mathbf1\\ \vdots\\ \mathbf1\end{pmatrix},
\end{equation}
with $\Gamma\tau+1$ copies in each lower stack. A YES instance has an indicator $c$ for which $B_0c-t_0$ is binary of weight at most $\tau$. In a NO instance,
\[
 |\supp(B_0z-wt_0)|>\Gamma\tau
\]
for every $z\in\Z^r$ and $w\in\Z$ with $w\ne0$. The same support bound holds for $z\in F^r$ and $0\ne w\in F$ over any field $F$.
\end{lemma}
\begin{proof}
For an exact-cover indicator $c$, we have $Sc=\mathbf1$, so the residual is $B_0c-t_0=(c,\mathbf0,\ldots,\mathbf0)^T$. Since $c\in\{0,1\}^r$ and $\|c\|_1\le\tau$, this residual is binary and has weight $\|B_0c-t_0\|_1=\|c\|_1\le\tau$, as required.

For soundness, first suppose $Sz-w\mathbf1\ne0$. Every repeated block then contains a nonzero entry, giving at least $\Gamma\tau+1$ nonzero coordinates. Otherwise $Sz=w\mathbf1$. Since $w\ne0$, every row has an incident set whose coefficient in $z$ is nonzero. These sets form an ordinary cover, and the NO promise gives $|\supp(z)|>\Gamma\tau$. Their coefficients need not be positive, and the argument is valid in every characteristic.
\end{proof}

We may repeat the entire pair $(B_0,t_0)$. If $B,t$ are obtained by stacking copies of this pair, let $s$ be $\tau$ times the number of copies. A YES residual is binary of weight at most $s$, while a NO instance satisfies
\begin{equation}\label{eq:cover-support}
 |\supp(Bz-wt)|>\Gamma s\qquad(w\ne0)
\end{equation}
over the integers and over every field. The number of copies will be chosen in the reduction below.

\subsection{The deterministic Reed--Solomon gadget}
Let $q$ be prime and let $k$ be an integer with $1\le k\le q/2$. The integer $k$ is the number of parity checks; its dependence on $q$ will be specified in Theorem~\ref{thm:gadget}. Index the coordinates by $a\in\F_q$ and define
\begin{equation}\label{eq:rs-definition}
 H_q(k)=(a^j)_{\substack{0\le j<k\\a\in\F_q}},
 \qquad
 L_{q,k}=\{z\in\Z^q:H_q(k)z\equiv0\pmod q\}.
\end{equation}
Thus $H_q(k)$ has $k$ rows and $q$ columns, and its first row is all ones, including at $a=0$.

\begin{lemma}[Reed--Solomon distance bounds]\label{lem:rs}
For $1\le k\le q/2$,
\[
 \lambda(L_{q,k})\ge2k,
 \qquad d_H(\ker_{\F_q}H_q(k))\ge k+1.
\]
\end{lemma}
\begin{proof}
Any set of at most $k$ columns is independent, by the Vandermonde determinant applied to its first as many rows. This proves the Hamming bound.

For the integer bound, suppose $0\ne z\in L_{q,k}$ has $\|z\|_1<2k\le q$. Form the multisets $S_+$ and $S_-$ of field elements appearing at positive and negative coordinates of $z$, respectively, with multiplicities $|z_a|$. The first parity check says that their cardinalities differ by a multiple of $q$. Their total is less than $q$, so their cardinalities are equal, say to $v<k$. If $v=0$, then $z=0$, a contradiction. Otherwise the other parity checks give equality of their power sums through degree $v$, since $v\le k-1$. Newton's identities recover their elementary symmetric functions from these power sums: the only divisions are by $1,\ldots,v$, which are nonzero in $\F_q$. Therefore
\[
 \prod_{a\in S_+}(X-a)=\prod_{a\in S_-}(X-a)
 \qquad\text{in }\F_q[X].
\]
This is equality of polynomial coefficients modulo $q$, with roots counted according to their multiplicities. Unique factorization in $\F_q[X]$ implies that the two root multisets coincide. But coordinates are indexed by distinct field elements, and the positive and negative coordinate supports are disjoint, so this is impossible when $v>0$. This proves the integer estimate, also given in~\cite[Theorem 14]{bp}.
\end{proof}

\Needspace{20\baselineskip}
\begin{samepage}
\begin{theorem}[Binary projection theorem {\cite{wan}}]\label{thm:gadget}
Fix rational constants
\begin{equation}\label{eq:gadget-exponents}
 \frac12<\alpha<1,\qquad 0<\delta<\xi<\frac{2\alpha-1}{4\alpha}.
\end{equation}
For every sufficiently large prime $q$, put
\begin{equation}\label{eq:gadget-parameters}
 k=\left\lfloor\frac{q^\xi}{2}\right\rfloor,
 \qquad h=\lfloor2\alpha k\rfloor,
 \qquad r=\lfloor q^\delta\rfloor.
\end{equation}
In deterministic $\operatorname{poly}(q)$ time one can construct an integer basis $A$ of $L_{q,k}$, an integer vector $x\in\Z^q$, and a coordinate projection $P:\Z^q\to\Z^r$ such that, for every $c\in\{0,1\}^r$, there is $z\in\Z^q$ with
\begin{equation}\label{eq:binary-projection}
 v=x+Az\in\{0,1\}^q,\qquad Pv=c,\qquad \|v\|_1=h.
\end{equation}
In particular, $\lambda(L_{q,k})\ge2k$ and $h\le2\alpha k$. All output lengths are polynomial in $q$. The lower cutoff on $q$ is effective and depends only on the fixed constants.
\end{theorem}
\end{samepage}

This is the parameter form of~\cite[Theorem 7.3 and its preceding binary-surjectivity proof]{wan}. In that paper, take $\epsilon=2\alpha-1$ and $\epsilon_1=\xi$. Its preceding proof establishes the projection property on the constant-weight binary subset itself. The additional bounds on the projection dimension have positive constant multiples of $k$ on their right sides; because $\delta<\xi$, they hold for all sufficiently large $q$. In particular, $\xi$ can be any sufficiently small positive rational. This freedom is used in Section~\ref{sec:hardness}.

The matrix $H_q(k)$ is the parity-check matrix defining $L_{q,k}$; the matrix $A$ is an integer basis matrix for this lattice. Constructing $A$ needs only finite-field linear algebra. Split $H_q(k)=[H_0\mid H_1]$ with $H_0$ an invertible $k$-column submatrix over $\F_q$, and choose an integer lift $V$ of $H_0^{-1}H_1$. In that coordinate order, take
\begin{equation}\label{eq:gadget-basis}
 A=\begin{pmatrix}qI_k&-V\\0&I_{q-k}\end{pmatrix}.
\end{equation}
To verify the basis claim, write an integer vector in two blocks as $(u,v)^T$, where $u\in\Z^k$ and $v\in\Z^{q-k}$. Multiplication by $H_0^{-1}$ over $\F_q$ gives
\[
 \begin{pmatrix}u\\v\end{pmatrix}\in L_{q,k}
 \quad\Longleftrightarrow\quad H_0u+H_1v\equiv0\pmod q
 \quad\Longleftrightarrow\quad u+Vv\equiv0\pmod q.
\]
The last congruence means $u+Vv=qz$ for some $z\in\Z^k$, and then
\[
 \begin{pmatrix}u\\v\end{pmatrix}
 =\begin{pmatrix}qz-Vv\\v\end{pmatrix}
 =A\begin{pmatrix}z\\v\end{pmatrix}.
\]
Conversely, every such integer combination satisfies the parity checks. Thus $L_{q,k}=A\Z^q$, and $\det A=q^k\ne0$ proves that the columns form a basis. In particular, $H_q(k)A\equiv0\pmod q$. Choosing the entries of $V$ among $0,\ldots,q-1$ gives $O(\log q)$-bit entries in $A$. The substantive external input is the binary coset projection. Any smaller collection of projected coordinates has the same property: extend the prescribed shorter binary vector to the full projection dimension, and apply \eqref{eq:binary-projection}.

\subsection{Combining the source and the gadget}
\begin{theorem}[Binary gap with modular soundness]\label{thm:base}
Fix rational $0<\varepsilon<1$. There is a deterministic polynomial-time reduction from an NP-hard promise problem to a lattice $L=L(\cB)\subseteq\Z^N$, a prime $q$, and $k=\lfloor q^\xi/2\rfloor\ge2$ such that
\begin{equation}\label{eq:base-gap}
\begin{aligned}
 \YES:&\quad \text{some }0\ne y\in L\cap\{0,1\}^N
                       \text{ has }\|y\|_1<(1+\varepsilon)k,\\
 \NO:&\quad \lambda(L)\ge2k,\qquad d_H(L\bmod q)\ge k+1.
\end{aligned}
\end{equation}
For any prescribed fixed $C,M>0$, the exponent $\xi$ can also be chosen to ensure $q>Ck^M$.
\end{theorem}
\begin{proof}
Choose
\[
 \alpha=\frac12+\frac\varepsilon8,
 \qquad \Gamma\in\Z,\quad\Gamma\ge\frac8\varepsilon,
 \qquad\delta=\frac\xi2,
\]
where $\xi>0$ is rational and satisfies \eqref{eq:gadget-exponents}. Start with a $\Gamma$-Gap Exact Set Cover instance with $r$ sets. Choose a sufficiently large polynomial-size prime $q$ for which
\begin{equation}\label{eq:base-balance}
 \lfloor q^\delta\rfloor\ge r,
 \qquad k\ge\frac{4\tau}{\varepsilon},
 \qquad k\ge2,
 \qquad2k\le q.
\end{equation}
The prime choice is justified at the end of the proof. Retain $r$ coordinates of the gadget projection and call the resulting map $P:\Z^q\to\Z^r$.

Form $B_0,t_0$ as in Lemma~\ref{lem:cover}. Set
\begin{equation}\label{eq:base-repetition}
 s=\tau\left\lceil\frac{2k}{\Gamma\tau}\right\rceil,
 \qquad\Gamma s\ge2k,
 \qquad s\le\frac{2k}{\Gamma}+\tau\le\frac{\varepsilon k}{2}.
\end{equation}
The last inequality uses $2/\Gamma\le\varepsilon/4$ and \eqref{eq:base-balance}. Stack $s/\tau=\lceil2k/(\Gamma\tau)\rceil$ copies of $B_0,t_0$ to obtain $B,t$. Output
\begin{equation}\label{eq:block-basis}
 \cB=\begin{pmatrix}BPA&BPx-t\\ A&x\end{pmatrix},
 \qquad L=L(\cB).
\end{equation}
This matrix has $q+1$ independent columns. Indeed, if $\cB(z,w)^T=0$, its lower block gives $Az+wx=0$. Substituting into the upper block gives $-wt=0$. The universe is nonempty, so $t\ne0$ and $w=0$; the invertibility of $A$ then gives $z=0$.

For completeness, let $c$ be an exact-cover indicator with $\|c\|_1\le\tau$. The gadget supplies a binary vector $v=x+Az$ of weight $h$ with $Pv=c$. Hence
\[
 y=\cB\begin{pmatrix}z\\1\end{pmatrix}
   =\begin{pmatrix}Bc-t\\v\end{pmatrix}
\]
is binary and nonzero. Moreover,
\[
 \|y\|_1\le s+h
 \le\frac{\varepsilon k}{2}+2\alpha k
 =\left(1+\frac{3\varepsilon}{4}\right)k
 <(1+\varepsilon)k.
\]
The reduction only needs the existence of this witness.

For soundness, write an arbitrary nonzero lattice vector as
\begin{equation}\label{eq:block-vector}
 \cB\begin{pmatrix}z\\w\end{pmatrix}
 =\begin{pmatrix}BPv-wt\\v\end{pmatrix},
 \qquad v=Az+wx.
\end{equation}
If $w=0$, then $z\ne0$ and $v=Az$ is a nonzero vector of $L_{q,k}$. Lemma~\ref{lem:rs} bounds its $\ell_1$ norm below by $2k$. If $w\ne0$, \eqref{eq:cover-support}, applied with the coefficient vector $Pv$, gives more than $\Gamma s\ge2k$ nonzero coordinates in the upper block. Each nonzero integer coordinate contributes at least one to the $\ell_1$ norm. Thus $\lambda(L)\ge2k$.

For the modular bound, reduce \eqref{eq:block-vector} modulo $q$ and suppose its residue is nonzero. If $w\equiv0\pmod q$, then $v\equiv Az\pmod q$. Since every column of $A$ lies in $L_{q,k}$, we have
\[
 H_q(k)v\equiv H_q(k)Az\equiv0\pmod q,
 \qquad v\bmod q\in\ker_{\F_q}H_q(k).
\]
This residue is nonzero: if $v\equiv0\pmod q$, the upper block $BPv-wt$ would vanish modulo $q$ as well, contrary to the chosen nonzero residue of the entire vector. Lemma~\ref{lem:rs} therefore supplies at least $k+1$ nonzero coordinates in the lower block. If $w\not\equiv0\pmod q$, the field version of \eqref{eq:cover-support} gives more than $\Gamma s\ge2k$ nonzero coordinates in the upper block. This proves $d_H(L\bmod q)\ge k+1$.

It remains to select $q$ effectively. All constants and exponents are fixed before the input is processed. Let $Q_0$ be the larger of a sufficiently large fixed cutoff and the ceiling of a sufficiently large fixed multiple of $(r+1)^{1/\delta}$. Bertrand's theorem provides a prime in $[Q_0,2Q_0]$. Scanning this interval and testing primality deterministically takes polynomial time in $r$; even trial division suffices because $Q_0$ is polynomial in $r$. Since $\xi=2\delta$, the resulting $k$ grows quadratically in $r+1$ outside the fixed cutoff. This ensures \eqref{eq:base-balance}, using $\tau\le r$. The cutoff covers the smaller inputs. Floors and ceilings of the fixed rational powers can be computed by exact integer power comparisons.

To impose $q>Ck^M$, choose $\xi$ still smaller so that $\xi M<1$, and increase the fixed cutoff. Then $k\le q^\xi$ gives $Ck^M\le Cq^{\xi M}<q$. The number of repetitions, the dimensions of \eqref{eq:block-basis}, and its entry bit lengths remain polynomial in the source encoding length.
\end{proof}

\begin{remark}\label{rem:binary-interface}
The one-copy binary $\ell_1$ gap can be made arbitrarily close to two: the NO lower bound is $2k$, and the YES weight is less than $(1+\varepsilon)k$. For Euclidean lengths, the YES witness satisfies $\|y\|_2^2=\|y\|_1<(1+\varepsilon)k$, while every nonzero NO vector satisfies $\|z\|_2^2\ge\|z\|_1\ge2k$. Thus these bounds give the Euclidean ratio
\[
 \sqrt{\frac{2}{1+\varepsilon}},
\]
which approaches $\sqrt2$. This recovers the range of fixed Euclidean factors below $\sqrt2$ in~\cite{wan}. The improvement to factors below two comes from the tensor square in the next section. Further amplification also uses the modular bound and the freedom to make $\xi$ small; both are properties of the construction just proved.
\end{remark}

\section{Two Euclidean tensor estimates}\label{sec:euclidean}
The integer $\ell_1$ bound will be used once, to obtain a Euclidean bound for the tensor square. All later amplification uses Euclidean minima and modular Hamming distance.

\begin{lemma}[Two-factor Euclidean bound]\label{lem:seed}
For nonzero integer-coordinate lattices $U,V$,
\begin{equation}\label{eq:seed}
 \lambda_2(U\otimes V)^2\ge\lambda(U)\lambda(V).
\end{equation}
\end{lemma}
\begin{proof}
Reshape $0\ne W\in U\otimes V$ as in \eqref{eq:matrix-form}. Let $a,b$ be its numbers of nonzero rows and columns, and put $S=\sum_{i,j}|W_{ij}|$. Every nonzero row has $\ell_1$ norm at least $\lambda(V)$; every nonzero column has $\ell_1$ norm at least $\lambda(U)$. Consequently,
\[
 S\ge a\lambda(V),\qquad S\ge b\lambda(U),
 \qquad S^2\ge ab\lambda(U)\lambda(V).
\]
The support of $W$ is contained in the rectangle of its $a$ active rows and $b$ active columns. Cauchy--Schwarz on these $ab$ entries gives $S^2\le ab\|W\|_F^2$. Dividing by $ab>0$ proves the bound for every nonzero $W$.
\end{proof}

For a NO lattice in Theorem~\ref{thm:base}, take $U=L^{\otimes2}$ and $D=k^2$. Lemmas~\ref{lem:seed} and~\ref{lem:hamming} give
\begin{equation}\label{eq:seed-scale}
 \lambda_2(U)^2\ge4D,
 \qquad d_H(U\bmod q)\ge(k+1)^2>D.
\end{equation}
\begin{remark*}[Preliminary Euclidean hardness below two]
The first inequality in \eqref{eq:seed-scale} already improves the Euclidean hardness range to every fixed approximation factor $1\le\rho<2$. Indeed, a NO instance satisfies $\lambda_2(L^{\otimes2})\ge2k$. In a YES instance, the binary witness has
\[
 \|y\otimes y\|_2=\|y\|_2^2=\|y\|_1<(1+\varepsilon)k.
\]
Given such a fixed $\rho$, choose rational $0<\varepsilon<1$ with $\rho(1+\varepsilon)<2$. Output the tensor-square basis $\cB^{\otimes2}$ and the rational threshold $d=(1+\varepsilon)k$. The YES minimum is less than $d$, and the NO minimum is at least $2k>\rho d$. Taking a tensor square preserves polynomial encoding length and deterministic polynomial construction time. Hence this gives deterministic NP-hardness for every fixed factor below $2$. The endpoint $2$ does not follow from these bounds alone; the subsequent amplification gives that endpoint and all larger fixed factors.
\end{remark*}
The stronger Hamming bound is available throughout. We normalize by $k^2$ because it gives the integer Euclidean coefficient $4$. Normalizing instead by $(k+1)^2$ would give the coefficient $4k^2/(k+1)^2<4$. Keeping the two scales separate could improve finite-stage bounds, but is unnecessary for arbitrary constant amplification.

We next estimate a tensor of small real matrix rank. We first need a factorization whose factors remain in the given lattices.

\begin{lemma}[Integer rank factorization]\label{lem:factorization}
If $0\ne W\in U\otimes V$ has real rank $r$, then $W=XY^T$, where the columns of $X$ and $Y$ are bases of rank-$r$ sublattices of $U$ and $V$, respectively.
\end{lemma}
\begin{proof}
Write $W=BZC^T$ with full-column-rank lattice bases $B,C$. Left inverses for $B,C$ over $\mathbb Q$ show that $\rank Z=\rank W=r$. The subgroup of $\Z^m$ generated by the rows of $Z$ has a $\Z$-basis of $r$ vectors. Put these basis vectors into the rows of an integer matrix $Q$. Every row of $Z$ is an integer combination of the rows of $Q$, so $Z=PQ$ for an integer matrix $P$. Since $Z$ has rank $r$, both $P$ and $Q$ have rank $r$. Set $X=BP$ and $Y=CQ^T$. Their columns are independent lattice vectors in $U$ and $V$, and $XY^T=BZC^T=W$. The subgroup basis used here follows from integer row reduction; its vectors need not be a subset of the original rows of $Z$.
\end{proof}

\Needspace{10\baselineskip}
\begin{lemma}[Low-rank Euclidean bound]\label{lem:lowrank}
Let $0\ne W\in U\otimes V$ have real rank $r$. If $r=1$, then
\[
 \|W\|_F^2\ge\lambda_2(U)^2\lambda_2(V)^2.
\]
If $r\ge2$, then
\begin{equation}\label{eq:lowrank}
 \|W\|_F^2\ge\frac9{4r}\lambda_2(U)^2\lambda_2(V)^2.
\end{equation}
In particular, for an integer $a\ge4$ and $1\le r\le2a$, the coefficient $9/(8a)$ is valid.
\end{lemma}
\begin{proof}
Use Lemma~\ref{lem:factorization}. The Gram matrices $G_X=X^TX$ and $G_Y=Y^TY$ are symmetric positive definite: for example, $z^TG_Xz=\|Xz\|_2^2>0$ for every nonzero real $z$. The $j$-th diagonal entry of $W^TW$ is $\sum_i W_{ij}^2$, so taking its trace sums the squares of all entries of $W$:
\[
 \|W\|_F^2=\tr(W^TW).
\]
Since $W=XY^T$, we have $W^T=YX^T$, and hence
\[
 \|W\|_F^2=\tr(YX^TXY^T)
 =\tr(X^TXY^TY)=\tr(G_XG_Y).
\]
Here the middle equality uses $\tr(AB)=\tr(BA)$ for rectangular matrices whose two products are defined and square.
To apply the arithmetic--geometric mean inequality, use the symmetric positive-definite matrix $G_Y^{1/2}G_XG_Y^{1/2}$. Here $G_Y^{1/2}$ is the positive-definite matrix square root: if
\[
 G_Y=Q\diag(\gamma_1,\ldots,\gamma_r)Q^T,\qquad\gamma_i>0,
\]
with $Q$ orthogonal, then
\[
 G_Y^{1/2}=Q\diag(\sqrt{\gamma_1},\ldots,\sqrt{\gamma_r})Q^T.
\]
This is a square root under matrix multiplication. The symmetric matrix above has trace $\tr(G_XG_Y)$ and determinant $\det G_X\det G_Y$. Applying the scalar arithmetic--geometric mean inequality to its $r$ positive eigenvalues gives
\begin{equation}\label{eq:determinant}
 \|W\|_F^2\ge r(\det G_X\det G_Y)^{1/r}
 =r\bigl(\det L(X)\det L(Y)\bigr)^{2/r}.
\end{equation}
Here $\det L(X)=\sqrt{\det G_X}$ is the covolume in its real span, and similarly for $Y$. This is the determinant estimate underlying~\cite[Claim 3.6]{hr}.

Let $v_r$ denote the volume of the unit Euclidean ball in dimension $r$. For any rank-$r$ lattice $\Lambda$, the open ball of radius $\lambda_2(\Lambda)/2$ injects into the quotient of its real span by $\Lambda$: the difference of two points in that ball has length strictly less than $\lambda_2(\Lambda)$. The ball's volume is therefore at most the quotient volume, which is the covolume. Thus
\[
 \det\Lambda\ge v_r\left(\frac{\lambda_2(\Lambda)}2\right)^r.
\]
Since $L(X)\subseteq U$ and $L(Y)\subseteq V$, their minima are at least those of $U,V$. Substituting the two covolume bounds in \eqref{eq:determinant} yields
\begin{equation}\label{eq:ball-bound}
 \|W\|_F^2\ge\frac{r v_r^{4/r}}{16}
                    \lambda_2(U)^2\lambda_2(V)^2.
\end{equation}

For completeness, the exact ball-volume formula gives
\begin{equation}\label{eq:ball-volume}
 v_r=\frac{\pi^{r/2}}{\Gamma(1+r/2)},
 \qquad v_{r+2}=\frac{2\pi}{r+2}v_r.
\end{equation}
The $\Gamma$ in this formula is the Gamma function. We claim that for every integer $r\ge2$,
\begin{equation}\label{eq:ball-elementary}
 v_r\ge\left(\frac{2\pi}{r}\right)^{r/2}
       >\left(\frac6r\right)^{r/2}.
\end{equation}
For $r=2$ the first inequality is equality. For $r=3$, $v_3=4\pi/3$ and
\[
 \frac{v_3}{(2\pi/3)^{3/2}}=\sqrt{\frac6\pi}>1.
\]
If the first inequality holds in dimension $r$, then \eqref{eq:ball-volume} gives
\[
\begin{aligned}
 v_{r+2}
 &\ge\frac{2\pi}{r+2}\left(\frac{2\pi}{r}\right)^{r/2}\\
 &=\left(\frac{2\pi}{r+2}\right)^{(r+2)/2}
     \left(\frac{r+2}{r}\right)^{r/2}
 \ge\left(\frac{2\pi}{r+2}\right)^{(r+2)/2}.
\end{aligned}
\]
The two base cases cover even and odd dimensions. The strict second inequality in \eqref{eq:ball-elementary} follows from $2\pi>6$. Therefore the coefficient in \eqref{eq:ball-bound} is at least $\pi^2/(4r)>9/(4r)$, proving \eqref{eq:lowrank}.

When $r=1$, the integer factorization is $W=xy^T$ for nonzero $x\in U,y\in V$, and $\|W\|_F^2=\|x\|_2^2\|y\|_2^2$. Finally, if $a\ge4$ and $1\le r\le2a$, the rank-one coefficient $1$ is at least $9/(8a)$, while $9/(4r)\ge9/(8a)$ for the other ranks.
\end{proof}

\section{Amplification by matrix rank and coordinate support}\label{sec:amplification}
The proof now combines the two soundness invariants. The integer $a$ measures the current squared Euclidean minimum relative to the modular distance scale. We use $2a+1$ columns to detect large modular rank, and the low-rank estimate for the remaining ranks at most $2a$.

\subsection{An integer collision inequality}
\begin{lemma}[Collision bound]\label{lem:collision}
For every integer $a\ge1$ and every list of $2a+1$ integers $z_1,\ldots,z_{2a+1}$,
\begin{equation}\label{eq:collision}
 \sum_{i=1}^{2a+1}z_i^2
 +a^2\#\{i<j:z_i=z_j\}
 \ge\frac{a(a+1)(2a+1)}3.
\end{equation}
\end{lemma}
\begin{proof}
Group the terms according to their integer value. If $z$ occurs $n_z$ times, it contributes
\[
 n_zz^2+a^2\binom{n_z}{2}
 =\sum_{j=0}^{n_z-1}(z^2+a^2j).
\]
Hence the left side is the sum of $2a+1$ selected entries, counted with multiplicity, from the multiset
\[
 \{z^2+a^2j:z\in\Z,\ j\ge0\}.
\]
The first occurrences of $-a,-a+1,\ldots,a$ each cost at most $a^2$. A repeated occurrence costs at least $a^2$, and a first occurrence outside this interval costs more than $a^2$. It follows that the $2a+1$ smallest costs have the same total as these $2a+1$ first occurrences. Ties at $a^2$ do not change that total. Any selection therefore costs at least
\[
 \sum_{z=-a}^{a}z^2
 =2\sum_{z=1}^{a}z^2
 =\frac{a(a+1)(2a+1)}3.
\]
This proves the inequality.
\end{proof}

\Needspace{23\baselineskip}
\subsection{One amplification phase}
\begin{samepage}
\begin{theorem}[Rank and support amplification]\label{thm:phase}
Fix integers $a\ge4$ and $n\ge1$, and put $F_a=a(a+1)/3$. Let $U\subseteq\Z^N$ be a nonzero lattice, let $D\ge1$ be an integer, and let $q$ be prime. Suppose
\begin{equation}\label{eq:phase-hypotheses}
 \lambda_2(U)^2\ge aD,
 \qquad d_H(U\bmod q)\ge D,
 \qquad q>(aD^n)^{a+1}.
\end{equation}
Define
\begin{equation}\label{eq:recurrence}
 c_1=a,\qquad c_{j+1}=\frac{a^3c_j}{a^3-F_a+c_j}.
\end{equation}
Then, for $1\le j\le n$,
\begin{equation}\label{eq:phase-conclusion}
 \lambda_2(U^{\otimes j})^2\ge c_jD^j.
\end{equation}
The coefficients have the explicit formula
\begin{equation}\label{eq:closed-form}
 c_j=\frac{F_a}{1+(F_a/a-1)(1-F_a/a^3)^{j-1}}
\end{equation}
and increase strictly from $a$ to $F_a$ as $j$ tends to infinity.
\end{theorem}
\end{samepage}

\begin{proof}
We begin by solving the recurrence, then prove the lattice bound by induction.

\paragraph{Solving the recurrence.}
Since $a\ge4$, we have $a<F_a<a^3$. All denominators in \eqref{eq:recurrence} are positive. Taking reciprocals and separating the numerator gives
\[
 \frac1{c_{j+1}}
 =\frac{a^3-F_a+c_j}{a^3c_j}
 =\left(1-\frac{F_a}{a^3}\right)\frac1{c_j}
   +\frac1{a^3}.
\]
The constant $1/F_a$ satisfies this same relation, because
\[
 \left(1-\frac{F_a}{a^3}\right)\frac1{F_a}
 +\frac1{a^3}=\frac1{F_a}.
\]
Subtracting this identity from the preceding one removes the constant term:
\begin{equation}\label{eq:centered-recurrence}
 \frac1{c_{j+1}}-\frac1{F_a}
 =\left(1-\frac{F_a}{a^3}\right)
   \left(\frac1{c_j}-\frac1{F_a}\right).
\end{equation}
Apply \eqref{eq:centered-recurrence} successively for the indices $1,\ldots,j-1$. Each step multiplies the difference by the same number, so $c_1=a$ gives
\begin{equation}\label{eq:reciprocal-formula}
 \frac1{c_j}-\frac1{F_a}
 =\left(1-\frac{F_a}{a^3}\right)^{j-1}
   \left(\frac1a-\frac1{F_a}\right).
\end{equation}
Adding $1/F_a$ and taking the reciprocal proves \eqref{eq:closed-form}. Both $1/a-1/F_a$ and $1-F_a/a^3$ are positive, and the latter is less than one. Thus the reciprocals decrease strictly to $1/F_a$, or equivalently
\begin{equation}\label{eq:coefficient-properties}
 a=c_1<c_2<\cdots<F_a,\qquad \lim_{j\to\infty}c_j=F_a.
\end{equation}
This limit concerns the numerical recurrence. For a fixed prime, the lattice assertion is only made through the prescribed finite order $n$ in \eqref{eq:phase-hypotheses}.

\paragraph{The inductive setup.}
The case $j=1$ is a hypothesis. Suppose \eqref{eq:phase-conclusion} holds at an index $j<n$. If it fails at $j+1$, there is a nonzero integer tensor which, in the orientation $U\otimes U^{\otimes j}$, is a matrix $W$ with
\begin{equation}\label{eq:short-tensor}
 \|W\|_F^2<c_{j+1}D^{j+1}.
\end{equation}
Its columns lie in $U$, and its transposed rows lie in $U^{\otimes j}$. Let $R$ be the number of nonzero rows, which we call active rows. Each has squared norm at least $c_jD^j$ by induction. Therefore
\begin{equation}\label{eq:active-rows}
 Rc_jD^j\le\|W\|_F^2,
 \qquad R<\frac{c_{j+1}}{c_j}D.
\end{equation}
We will contradict \eqref{eq:short-tensor} for both possible ranges of the rank of $W\bmod q$.

\paragraph{Selecting short columns when the modular rank is large.}
Suppose first that $\rank_{\F_q}(W\bmod q)\ge2a+1$. We claim that $W$ has columns $u_1,\ldots,u_{2a+1}$ whose reductions are independent and whose squared norms satisfy
\begin{equation}\label{eq:short-columns}
 \|u_i\|_2^2\le\frac{\|W\|_F^2}{D^j}<c_{j+1}D
 \qquad(1\le i\le2a+1).
\end{equation}
To see this, suppose fewer than $2a+1$ such columns have been chosen. Write $\bar W=W\bmod q$ for this selection argument, and let $S$ be the span of the selected column reductions. Since $\rank_{\F_q}\bar W\ge2a+1$, the space $S$ is a proper subspace of the column space of $\bar W$. Choose a column $\bar w$ outside $S$. Extend a basis of $S$, first by $\bar w$ and then to a basis of $\F_q^N$. Defining a functional to be zero on the basis of $S$, one on $\bar w$, and zero on the remaining basis vectors gives a linear functional $\varphi$ with
\[
 \varphi|_S=0,\qquad \varphi(\bar w)=1.
\]
In standard coordinates, write $\varphi(x)=\sum_{i=1}^N\beta_i x_i$. Applying $\varphi$ to every column of $\bar W$ produces the row vector
\[
 \beta^T\bar W=\sum_{i=1}^N\beta_i\,\operatorname{row}_i(\bar W).
\]
This is a linear combination of the rows, and the resulting vector is nonzero: its coordinate at the chosen column $\bar w$ equals $\varphi(\bar w)=1$. Each row lies in the linear code
\[
 U^{\otimes j}\bmod q=(U\bmod q)^{\otimes j},
\]
so their nonzero combination lies in the same code and has at least $D^j$ nonzero coordinates by Lemma~\ref{lem:hamming}. At each such coordinate, $\varphi$ is nonzero on the corresponding column. That column therefore lies outside $S$, where $\varphi$ vanishes. These at least $D^j$ integer columns have total squared norm at most $\|W\|_F^2$, so at least one has squared norm at most $\|W\|_F^2/D^j$. Choose it. Its reduction extends the independent family, and repeating this procedure proves \eqref{eq:short-columns}.

\paragraph{Counting coordinate collisions.}
For distinct selected columns $u_i,u_l$, independence implies that $(u_i-u_l)\bmod q$ is nonzero. Since $u_i-u_l\in U$, it differs from zero on at least $D$ coordinates modulo $q$. All such coordinates lie among the $R$ active rows. Thus $R\ge D$, and the pair $u_i,u_l$ can agree as integers on at most $R-D$ active rows. Here we only use that integer equality implies equality modulo $q$.

We next count all occurrences of a pair of selected columns agreeing in an active row. Number the active rows $1,\ldots,R$ for this count. Summing first over rows or first over pairs counts the same occurrences, so
\[
\begin{aligned}
 \sum_{r=1}^R\#\{i<l:(u_i)_r=(u_l)_r\}
 &=\sum_{1\le i<l\le2a+1}
       \#\{1\le r\le R:(u_i)_r=(u_l)_r\}\\
 &\le\binom{2a+1}{2}(R-D).
\end{aligned}
\]
For each fixed pair, the bound $R-D$ already counts its agreeing rows among all $R$ active rows. Thus there is no additional factor of $R$. In an individual row there may be as many as $\binom{2a+1}{2}$ agreeing pairs; the preceding bound controls their total across rows.

Apply Lemma~\ref{lem:collision} to the $2a+1$ selected entries in each active row, and sum. The right side of that lemma contributes $(2a+1)F_a$ per row, while the total collision count is bounded by the displayed inequality. Hence
\begin{align}
 \sum_{i=1}^{2a+1}\|u_i\|_2^2
 &\ge(2a+1)F_aR-a^2\binom{2a+1}{2}(R-D)\notag\\
 &=(2a+1)\bigl[a^3D-(a^3-F_a)R\bigr].
 \label{eq:collision-sum}
\end{align}
The second line uses $\binom{2a+1}{2}=a(2a+1)$. Since $a^3-F_a>0$, the strict upper bound on $R$ in \eqref{eq:active-rows} now gives
\[
 \sum_{i=1}^{2a+1}\|u_i\|_2^2
 >(2a+1)\left[a^3-(a^3-F_a)\frac{c_{j+1}}{c_j}\right]D.
\]
The bracket equals $c_{j+1}$. Indeed, multiplying \eqref{eq:recurrence} by its denominator and dividing by $c_j$ gives
\[
 (a^3-F_a)\frac{c_{j+1}}{c_j}+c_{j+1}=a^3.
\]
This also explains the choice of the recurrence: it makes the collision lower bound match the proposed next coefficient. We have shown that the sum of the $2a+1$ squared column norms is greater than $(2a+1)c_{j+1}D$, contradicting \eqref{eq:short-columns}. Therefore
\begin{equation}\label{eq:mod-rank}
 \rank_{\F_q}(W\bmod q)\le2a.
\end{equation}

\paragraph{Transferring the rank bound to the reals.}
If $\rank_{\R}W\ge2a+1$, some square submatrix $W_0$ of order $2a+1$ has a nonzero integer determinant. By \eqref{eq:mod-rank}, this determinant vanishes modulo $q$, so $q$ divides it. We now bound its absolute value. Let $w_1,\ldots,w_{2a+1}$ be the columns of $W_0$. Hadamard's inequality states that the absolute determinant is at most the product of the Euclidean column lengths:
\[
 |\det W_0|\le\prod_{i=1}^{2a+1}\|w_i\|_2.
\]
The arithmetic--geometric mean inequality, applied to the $2a+1$ nonnegative numbers $\|w_i\|_2^2$, gives
\[
 \prod_{i=1}^{2a+1}\|w_i\|_2^2
 \le\left(\frac{\sum_{i=1}^{2a+1}\|w_i\|_2^2}{2a+1}\right)^{2a+1}.
\]
Taking square roots and using $\sum_i\|w_i\|_2^2=\|W_0\|_F^2\le\|W\|_F^2$, we obtain
\[
 |\det W_0|
 \le\left(\frac{\|W_0\|_F^2}{2a+1}\right)^{(2a+1)/2}
 \le\left(\frac{\|W\|_F^2}{2a+1}\right)^{a+1/2}.
\]
By \eqref{eq:short-tensor}, \eqref{eq:coefficient-properties}, $j+1\le n$, and $D\ge1$, we have $\|W\|_F^2<F_aD^n$. Retaining the denominator $2a+1$ therefore yields the simpler bound
\[
 \frac{\|W\|_F^2}{2a+1}
 <\frac{F_aD^n}{2a+1}
 =\frac{a(a+1)}{3(2a+1)}D^n
 <aD^n.
\]
Since $aD^n>1$, the determinant satisfies
\[
 0<|\det W_0|<(aD^n)^{a+1/2}<(aD^n)^{a+1}<q.
\]
A nonzero integer of absolute value less than $q$ cannot be divisible by $q$. This contradiction proves
\begin{equation}\label{eq:real-rank}
 1\le\rank_{\R}W\le2a.
\end{equation}
\paragraph{The low-rank contradiction.}
By \eqref{eq:real-rank}, the nonzero matrix $W$ has real rank at most $2a$, so Lemma~\ref{lem:lowrank} applies with coefficient $9/(8a)$. Combining it with the hypothesis on $U$ and the induction hypothesis gives
\begin{equation}\label{eq:lowrank-induction}
 \|W\|_F^2\ge\frac9{8a}\lambda_2(U)^2\lambda_2(U^{\otimes j})^2
 \ge\frac9{8a}(aD)(c_jD^j)
 =\frac98c_jD^{j+1}.
\end{equation}
It remains to check that $(9/8)c_j>c_{j+1}$. Directly from \eqref{eq:recurrence},
\[
 \frac{c_{j+1}}{c_j}
 =\frac1{1-F_a/a^3+c_j/a^3}
 <\frac1{1-F_a/a^3}.
\]
For $a\ge4$,
\[
 \frac{F_a}{a^3}=\frac{a+1}{3a^2}\le\frac5{48},
\]
because the equivalent inequality is
\[
 5a^2-16a-16=(a-4)(5a+4)\ge0.
\]
Consequently,
\begin{equation}\label{eq:coefficient-ratio}
 \frac{c_{j+1}}{c_j}<\frac{48}{43}<\frac98.
\end{equation}
Combining \eqref{eq:lowrank-induction} and \eqref{eq:coefficient-ratio} contradicts \eqref{eq:short-tensor}. This completes the induction and the proof.
\end{proof}

\subsection{Repeated amplification of tensor blocks}
\begin{proposition}[Arbitrarily large constant coefficients]\label{prop:schedule}
For every real $b_0>0$, there are an integer $b>b_0$ with $b\ge4$ and an even integer $T\ge2$, depending only on $b_0$, with the following property. If an integer-coordinate lattice $L$ and a prime $q$ satisfy
\begin{equation}\label{eq:schedule-hypotheses}
 \lambda(L)\ge2k,\qquad d_H(L\bmod q)\ge k,\qquad k\ge2,
 \qquad q>(bk^T)^{b+1},
\end{equation}
then
\begin{equation}\label{eq:schedule-conclusion}
 \lambda_2(L^{\otimes T})^2\ge bk^T.
\end{equation}
Given a rational upper bound on $b_0$, a finite rational-arithmetic algorithm computes such $b,T$ and the tensor schedule.
\end{proposition}
\begin{proof}
Start with $U=L^{\otimes2}$, $D=k^2$, and $a=4$. Lemmas~\ref{lem:seed} and~\ref{lem:hamming} give
\[
 \lambda_2(U)^2\ge aD,\qquad d_H(U\bmod q)\ge D.
\]
Suppose a phase starts with these two inequalities for an integer $a\ge4$. Choose an integer $a'$ with $a<a'<F_a$. The increasing convergence in Theorem~\ref{thm:phase} gives a finite $n$ with $c_n>a'$. Subject to the prime condition, the theorem and Lemma~\ref{lem:hamming} yield
\[
 \lambda_2(U^{\otimes n})^2\ge c_nD^n>a'D^n,
 \qquad d_H(U^{\otimes n}\bmod q)\ge D^n.
\]
These are the same two hypotheses with $U^{\otimes n}$ in place of $U$, $D^n$ in place of $D$, and $a'$ in place of $a$. We may therefore repeat the argument. This operation only regroups the existing tensor factors; it introduces no rescaling or new gadget.

The integer targets can be chosen as
\begin{equation}\label{eq:targets}
 4,\ 6,\ 12,\ 24,\ 48,\ldots.
\end{equation}
Indeed, $F_4=20/3>6$, and $F_a=a(a+1)/3>2a$ for every $a\ge6$. Stop at an integer $b>b_0$. If $b_0<4$, the starting tensor square suffices with $b=4,T=2$. Otherwise, the total tensor order is twice the product of the finitely many phase orders. Denote it by $T$; it is even and depends only on the target bound.

We now verify all the prime conditions at once. In a phase starting with $U=L^{\otimes t}$, the scale is $D=k^t$. If its phase order is $n$, then $tn\le T$. Also $a\le b$, so
\[
 (aD^n)^{a+1}
 =(ak^{tn})^{a+1}
 \le(bk^T)^{b+1}<q.
\]
Every application of Theorem~\ref{thm:phase} is therefore valid, and the last phase proves \eqref{eq:schedule-conclusion}.

For effectivity, follow the targets in \eqref{eq:targets}. At each phase, start from $c_1=a$ and iterate the rational recurrence \eqref{eq:recurrence} until its value exceeds the next target. The explicit limit proves that this process terminates. Only finitely many phases are needed, and all of this computation depends on the fixed target bound, not on the input instance.
\end{proof}

\section{Arbitrary constant hardness in Euclidean space}\label{sec:hardness}
The remaining issue is compatibility between the large-prime hypothesis in Proposition~\ref{prop:schedule} and the deterministic gadget. We choose the tensor schedule before the gadget exponents, so the prime need only exceed a fixed power of $k$.

\subsection{Choosing the prime after the tensor schedule}
Fix $b,T$ from Proposition~\ref{prop:schedule} and rational $0<\varepsilon<1$. In Theorem~\ref{thm:base}, set $\alpha=1/2+\varepsilon/8$ and choose rational exponents with
\begin{equation}\label{eq:exponent-choice}
 0<\xi<\min\left\{\frac{2\alpha-1}{4\alpha},\frac1{2T(b+1)}\right\},
 \qquad \delta=\frac\xi2.
\end{equation}
The first restriction is the gadget condition. For the second, $k\le q^\xi$ gives
\begin{equation}\label{eq:prime-bound}
 (bk^T)^{b+1}
 \le b^{b+1}q^{\xi T(b+1)}
 < b^{b+1}q^{1/2}<q
\end{equation}
for all sufficiently large $q$. The fixed cutoff $q>b^{2(b+1)}$ ensures the last inequality. Increase the cutoff to meet the gadget and balance conditions as well.

For a source instance with $r$ sets, choose $Q_0$ as in the proof of Theorem~\ref{thm:base}, above this fixed cutoff and of order $(r+1)^{1/\delta}$ once $r$ is large. A prime in $[Q_0,2Q_0]$ can be found deterministically in polynomial time. The new lower bound on $q$ changes only the fixed constants and degree of that polynomial.

The order of choices is therefore: the approximation factor; the finite schedule $b,T$; the YES slack $\varepsilon$; the gadget constants $\alpha,\Gamma,\xi,\delta$; and finally the prime as a function of the source size. In particular, the schedule never depends on the eventual prime or on the source length.

\subsection{Completing the reduction}
\begin{proof}[Proof of Theorem~\ref{thm:main}]
Fix $\rho>1$. Choose a rational upper bound on $\rho^2$ and use Proposition~\ref{prop:schedule} to obtain an integer $b>\rho^2$ and an even $T$. Since $T$ is fixed, we can choose rational $0<\varepsilon<1$ small enough that
\begin{equation}\label{eq:epsilon-choice}
 \rho^2(1+\varepsilon)^T<b.
\end{equation}
To make this choice effective without computing $\rho$, it suffices to impose the inequality with the chosen rational upper bound in place of $\rho^2$.

Construct the base instance in Theorem~\ref{thm:base} with the exponent and prime choices just proved. Output
\begin{equation}\label{eq:output}
 \left(\cB^{\otimes T},\quad d=((1+\varepsilon)k)^{T/2}\right).
\end{equation}
The threshold $d$ is rational exactly, because $T/2$ is an integer.

For a YES input, there is a nonzero binary $y\in L$ with $\|y\|_1<(1+\varepsilon)k$. Its tensor power is nonzero and binary, and hence
\[
 \|y^{\otimes T}\|_2^2=\|y\|_1^T
 <((1+\varepsilon)k)^T=d^2.
\]
For a NO input, the base lattice satisfies both soundness conditions of Proposition~\ref{prop:schedule}, and \eqref{eq:prime-bound} supplies the required prime. Therefore
\[
 \lambda_2(L^{\otimes T})^2\ge bk^T
 >\rho^2((1+\varepsilon)k)^T=\rho^2d^2.
\]
This proves the two output promises.

Finally, $\cB$ has polynomially many rows and columns and polynomial entry bit lengths. Its tensor power has $N^T$ rows and $(q+1)^T$ columns; its entries are products of $T$ original entries. The threshold in \eqref{eq:output} also has polynomial bit length. All construction parameters and $T$ are fixed once $\rho$ is fixed. The entire many-one reduction is consequently deterministic and polynomial-time.
\end{proof}

The polynomial degree may be large. The theorem asserts a polynomial-time reduction for each fixed approximation factor; it does not give a polynomial bound uniform in a factor growing with the input.

\section{All fixed finite \texorpdfstring{$\ell_p$}{lp} norms}\label{sec:norms}
For $p\ge2$, the same tensor construction admits a direct norm comparison, given at the end of this section. To include $1\le p<2$, and to give a transfer valid for every fixed finite $p$, we construct an explicit integer sign embedding. Norm embeddings have been used in lattice reductions in~\cite{rr}; the particular construction and all estimates needed here are proved below.

\subsection{An explicit sign embedding}
\begin{theorem}[Enumerated sign embedding]\label{thm:embedding}
Fix $1\le p<\infty$ and an integer $r\ge\max\{2,\lceil p/2\rceil\}$. For every $N\ge2$, a deterministic algorithm constructs $R\in\{-1,1\}^{M\times N}$, where $M=O(N^{2r})$, such that for every $x\in\R^N$,
\begin{equation}\label{eq:embedding}
 \frac{M^{1/p}}{\sqrt3}\|x\|_2
 \le\|Rx\|_p
 \le\sqrt{2r-1}\,M^{1/p}\|x\|_2.
\end{equation}
The algorithm is polynomial-time for fixed $r$, and $R^TR=MI_N$, so $R$ is injective. When $N$ is a power of two, the construction has exactly $M=N^{2r}$ rows. For $1\le p\le2$, the upper constant $\sqrt{2r-1}$ can be replaced by $1$.
\end{theorem}
\begin{proof}
We construct every row explicitly. Averages in the proof are uniform averages over this finite list, counting repeated rows according to their multiplicity. They are a way to evaluate the deterministic construction.

\paragraph{The sign rows.}
Let $Q=2^h$ be the smallest power of two at least $N$, so $h\ge1$ and
\[
 N\le Q<2N.
\]
Choose distinct $t_1,\ldots,t_N\in\F_Q$. We use the absolute field trace
\[
 \Tr(z)=\Tr_{\F_Q/\F_2}(z)
       =\sum_{i=0}^{h-1}z^{2^i}.
\]
This is an $\F_2$-linear map into $\F_2$: additivity follows from the characteristic-two power identities, and $\Tr(z)^2=\Tr(z)$ follows from $z^{2^h}=z$. The displayed polynomial is nonzero and has degree $2^{h-1}<Q$, so it cannot vanish on all of $\F_Q$. The trace is therefore surjective onto $\F_2$. Its kernel has $Q/2$ elements, and each bit has exactly $Q/2$ preimages.

Enumerate all $Q^{2r}$ coefficient tuples of the polynomials
\[
 f(t)=a_0+a_1t+\cdots+a_{2r-1}t^{2r-1},
 \qquad a_i\in\F_Q.
\]
For each coefficient tuple, include the row
\begin{equation}\label{eq:sign-row}
 \bigl((-1)^{\Tr(f(t_1))},\ldots,(-1)^{\Tr(f(t_N))}\bigr)
\end{equation}
in $R$. Thus $M=Q^{2r}<(2N)^{2r}$, which is $O(N^{2r})$ for fixed $r$; if $N$ is a power of two, then $Q=N$ and $M=N^{2r}$. Distinct coefficient tuples may yield the same row; keeping those repetitions makes averaging over rows exactly the same as averaging over all coefficient tuples.

\paragraph{Independence of small sets of coordinates.}
Let $s=(s_1,\ldots,s_N)$ be a uniformly indexed row of $R$. Choose $1\le j\le\min\{2r,N\}$ distinct coordinate positions $i_1,\ldots,i_j$. The evaluation map
\[
 \F_Q^{2r}\longrightarrow\F_Q^j,
 \qquad(a_0,\ldots,a_{2r-1})
 \longmapsto\bigl(f(t_{i_1}),\ldots,f(t_{i_j})\bigr)
\]
is a surjective linear map. Indeed, Lagrange interpolation produces a polynomial of degree at most $j-1<2r$ for any prescribed tuple of evaluations. Its kernel has dimension $2r-j$, so every prescribed tuple has exactly $Q^{2r-j}$ preimages. The $j$ evaluations are consequently independent and uniform on $\F_Q$ under the full enumeration.

Applying the trace coordinatewise gives independent unbiased bits, because each prescribed bit has $Q/2$ preimages. The map $b\mapsto(-1)^b$ is a bijection from $\F_2$ to $\{-1,1\}$, so the corresponding \emph{$j$ signs} in \eqref{eq:sign-row} are independent and unbiased. More explicitly, for every $(\varepsilon_1,\ldots,\varepsilon_j)\in\{-1,1\}^j$,
\[
 \Pr(s_{i_1}=\varepsilon_1,\ldots,s_{i_j}=\varepsilon_j)
 =\frac{Q^{2r-j}(Q/2)^j}{Q^{2r}}=2^{-j}.
\]
This is $2r$-wise independence: every subset of at most $2r$ of the available coordinates has the distribution of independent signs. If $N\ge2r$, the bound on $j$ is $j\le2r$; if $N<2r$, it is $j\le N$, and all $N$ signs are jointly independent. When $N>2r$, full joint independence of all $N$ signs is unnecessary; only the stated independence of small subsets will be used.

Pairwise independence already gives
\begin{equation}\label{eq:embedding-gram}
 R^TR=MI_N.
\end{equation}
To verify this identity entry by entry, write
\[
 (R^TR)_{i\ell}=\sum_{t=1}^M R_{ti}R_{t\ell}.
\]
For $i=\ell$, each summand is $1$, so the sum is $M$. For $i\ne\ell$, pairwise independence says that $(R_{ti},R_{t\ell})$ takes each of the four values $(1,1),(1,-1),(-1,1),(-1,-1)$ on exactly $M/4$ indexed rows. Thus
\[
 (R^TR)_{i\ell}=\frac M4(1-1-1+1)=0.
\]
This proves \eqref{eq:embedding-gram} by exact counting, including any repeated rows.

\paragraph{The lower norm estimate.}
Fix $x\in\R^N$ and set $X=\sum_i s_ix_i$ for the uniformly indexed row $s$ above. As a function of the selected row index, $X(t)=(Rx)_t$ for $1\le t\le M$. Thus, by the definition of a uniform average,
\[
 \E|X|^p=\frac1M\sum_{t=1}^M
       \left|\sum_{i=1}^N R_{ti}x_i\right|^p
       =\frac1M\|Rx\|_p^p.
\]
Consequently,
\begin{equation}\label{eq:finite-moment}
 \|Rx\|_p^p=M\E|X|^p.
\end{equation}
This identity requires no independence; it is simply the sum defining the $\ell_p$ norm written as an average. We now use independence to compare this finite moment with $\|x\|_2$. Since $\E s_i^2=1$ and $\E s_is_j=0$ for $i\ne j$,
\[
 \E X^2=\sum_i x_i^2\E s_i^2
             +2\sum_{i<j}x_ix_j\E s_is_j
          =\sum_i x_i^2=\|x\|_2^2.
\]
Four-wise independence, available since $r\ge2$, gives
\begin{equation}\label{eq:fourth-moment}
 \E X^4=\sum_i x_i^4+6\sum_{i<j}x_i^2x_j^2
 =3\left(\sum_i x_i^2\right)^2-2\sum_i x_i^4
 \le3\|x\|_2^4.
\end{equation}
In these expansions, a term with any sign appearing an odd number of times has average zero. In the fourth moment, the remaining terms either use one index four times, or two indices twice each. For a fixed pair $i<j$, the latter case has $\binom42=6$ arrangements, which explains the coefficient $6$ in \eqref{eq:fourth-moment}. Each expansion uses at most four distinct indices, so this reasoning also applies when $N<4$.

For scalar functions $f,g:\{1,\ldots,M\}\to\R$, H\"older's inequality for the uniform average over row indices states that
\[
 \E|fg|\le(\E|f|^u)^{1/u}(\E|g|^v)^{1/v},
 \qquad u,v>1,\qquad \frac1u+\frac1v=1.
\]
Here the product is pointwise, so $\E|fg|=M^{-1}\sum_{t=1}^M|f(t)g(t)|$. Take $f(t)=|X(t)|^{2/3}$, $g(t)=|X(t)|^{4/3}$ and the conjugate exponents $u=3/2$, $v=3$. Then $fg=|X|^2$, $|f|^{3/2}=|X|$, and $|g|^3=|X|^4$, giving
\[
 \E X^2\le(\E|X|)^{2/3}(\E X^4)^{1/3}.
\]
For $x\ne0$, raising to the power $3/2$ and rearranging gives
\[
 \E|X|\ge\frac{(\E X^2)^{3/2}}{(\E X^4)^{1/2}}
 \ge\frac{\|x\|_2}{\sqrt3}.
\]
For $x=0$, every required inequality holds directly. We also use monotonicity of the means on this probability space. For $0<a<b$, apply H\"older to $|X|^a\cdot1$ with exponents $b/a$ and $b/(b-a)$ to obtain
\[
 \E|X|^a\le(\E|X|^b)^{a/b},
 \qquad (\E|X|^a)^{1/a}\le(\E|X|^b)^{1/b}.
\]
The factor involving the constant function is $1$ because $\E1=1$. Equality of the exponents gives equality of the means. Taking $a=1$, $b=p$ when $p>1$, and using equality when $p=1$, gives
\begin{equation}\label{eq:moment-lower}
 (\E|X|^p)^{1/p}\ge\E|X|\ge\frac{\|x\|_2}{\sqrt3}.
\end{equation}
This establishes a lower bound that includes $p=1$.

\paragraph{The upper norm estimate.}
Since $p\le2r$, monotonicity of the means gives
\[
 (\E|X|^p)^{1/p}\le(\E X^{2r})^{1/(2r)}.
\]
Every term in the expansion of $X^{2r}$ uses at most $2r$ distinct signs, so its average agrees with that for fully independent signs. Only terms in which each index has even multiplicity survive. Each surviving monomial is nonnegative and admits a pairing of its $2r$ positions in which paired positions have the same coordinate index.

A perfect pairing partitions the $2r$ labeled positions into $r$ unordered pairs. Position $1$ has $2r-1$ possible partners. After removing this pair, take the smallest remaining position; it has $2r-3$ possible partners. Continue in this way until the last two positions must be paired. This procedure counts each pairing exactly once, because choosing the smallest unused position fixes the order in which its pairs are formed. The number of perfect pairings is therefore
\[
 (2r-1)(2r-3)\cdots3\cdot1=(2r-1)!!.
\]
For one fixed pairing, summing over an arbitrary index for each pair gives
\[
 \left(\sum_i x_i^2\right)^r=\|x\|_2^{2r}.
\]
Summing this over all pairings counts every surviving term at least once. All terms are nonnegative, so overcounting gives an upper bound:
\[
 \E X^{2r}\le(2r-1)!!\,\|x\|_2^{2r}
 \le(2r-1)^r\|x\|_2^{2r}.
\]
Consequently,
\begin{equation}\label{eq:moment-upper}
 (\E|X|^p)^{1/p}\le\sqrt{2r-1}\,\|x\|_2.
\end{equation}
When $p\le2$, the simpler bound $(\E|X|^p)^{1/p}\le(\E X^2)^{1/2}=\|x\|_2$ improves the upper constant to $1$. Combining \eqref{eq:finite-moment}, \eqref{eq:moment-lower}, and \eqref{eq:moment-upper} proves \eqref{eq:embedding}. Equation~\eqref{eq:embedding-gram} gives injectivity: if $Rx=0$, then $M\|x\|_2^2=x^TR^TRx=0$.

\paragraph{Deterministic construction time.}
The number of rows is $Q^{2r}=O(N^{2r})$, and each has $N$ signs. The field $\F_Q$ can be constructed by finding a monic irreducible polynomial over $\F_2$ of degree $h=\log_2Q$. Even enumerating all monic candidates of this degree and testing divisibility by all monic polynomials of positive degrees at most $h/2$ costs polynomial time in $Q$: there are $Q$ candidates and fewer than $2\sqrt Q$ possible trial divisors. A reducible polynomial has an irreducible factor of degree at most half its degree, so this test is sufficient. An irreducible polynomial of every positive degree exists over $\F_2$. Field arithmetic and polynomial evaluation are then explicit, and the trace uses $h-1$ successive squarings and $h-1$ additions. Since $Q<2N$ and $r$ is fixed, constructing the entire matrix takes deterministic polynomial time in $N$.
\end{proof}

\subsection{Transferring the promise and the rational radius}
\begin{proposition}[Deterministic finite-norm transfer]\label{prop:transfer}
For every fixed finite effective $p\ge1$, deterministic arbitrary-constant hardness of Euclidean GapSVP implies deterministic arbitrary-constant hardness of GapSVP$_p$.
\end{proposition}
\begin{proof}
Fix the desired factor $\rho>1$ and an integer $r\ge\max\{2,\lceil p/2\rceil\}$. Use Euclidean hardness at a fixed factor
\begin{equation}\label{eq:transfer-gap}
 \gamma>2\rho\sqrt{3(2r-1)}.
\end{equation}
Given a Euclidean instance $(B,d)$ with $B\in\Z^{N\times n}$, append a zero row to $B$ if necessary to ensure $N\ge2$; this preserves every vector norm and the lattice rank. Construct $R$ from Theorem~\ref{thm:embedding} and output the integer basis $B'=RB$. Since $R$ is injective,
\[
 L(B')=RL(B),\qquad\rank B'=\rank B.
\]
The norm inequalities hold simultaneously for every real vector, and hence for every lattice vector. Taking minima gives
\begin{equation}\label{eq:transfer-minima}
 \frac{M^{1/p}}{\sqrt3}\lambda_2(L(B))
 \le\lambda_p(L(B'))
 \le\sqrt{2r-1}\,M^{1/p}\lambda_2(L(B)).
\end{equation}
Compute a rational $d'>0$ with
\begin{equation}\label{eq:transfer-radius}
 \sqrt{2r-1}\,M^{1/p}d
 \le d'\le2\sqrt{2r-1}\,M^{1/p}d.
\end{equation}
The fixed-$p$ convention and ordinary square-root approximation provide this upper approximation in polynomial time and with polynomial bit length. For example, approximate each of the two positive factors from above to relative error at most $1/4$; their product is within a factor $25/16<2$.

If the Euclidean input is YES, \eqref{eq:transfer-minima} implies $\lambda_p(L(B'))\le d'$. If it is NO, the same inequality gives
\[
 \lambda_p(L(B'))
 >\frac{M^{1/p}}{\sqrt3}\gamma d
 >2\rho\sqrt{2r-1}\,M^{1/p}d
 \ge\rho d',
\]
where the second step uses \eqref{eq:transfer-gap} and the last uses \eqref{eq:transfer-radius}. Thus the output has the required promise.

The matrix $R$ has polynomially many rows and entries of absolute value one. Each entry of $RB$ is a signed sum of $N$ entries of $B$, increasing the entry bit length by at most $O(\log N)$. Matrix multiplication and radius approximation take deterministic polynomial time. This proves the reduction.
\end{proof}

\begin{proof}[Proof of Theorem~\ref{thm:allp}]
Apply Theorem~\ref{thm:main} at the larger fixed Euclidean factor required by Proposition~\ref{prop:transfer}, and compose the two reductions. Their dependence on $p$ and $\rho$ is through fixed constants only.
\end{proof}

\subsection{A direct extension for \texorpdfstring{$p\ge2$}{p >= 2}}
The tensor instances permit a shorter argument when $p\ge2$. For every integer vector $z$, the coordinate inequality $|z_i|^p\ge|z_i|^2$ gives $\|z\|_p^p\ge\|z\|_2^2$. Therefore Proposition~\ref{prop:schedule} gives, in NO instances,
\[
 \lambda_p(L^{\otimes T})^p\ge bk^T.
\]
In YES instances, the binary witness gives
\[
 \|y^{\otimes T}\|_p^p=\|y\|_1^T<(1+\varepsilon)^Tk^T.
\]
Choose $b>\rho^p$, its fixed even tensor order $T$, and then rational $\varepsilon>0$ so small that $\rho^p(1+\varepsilon)^T<b$. To encode the radius explicitly, choose a fixed rational $\eta>0$ with
\[
 \rho^p(1+\eta)^p(1+\varepsilon)^T<b.
\]
An upper rational approximation $d$ to $((1+\varepsilon)k)^{T/p}$ within relative factor $1+\eta$ then satisfies the YES bound and leaves $\lambda_p(L^{\otimes T})>\rho d$ in NO instances. Such an approximation is covered by the fixed-$p$ convention. The exponent and prime selection in Section~\ref{sec:hardness} is unchanged. For $p<2$, the coordinate inequality no longer supplies this lower bound, so the sign embedding provides the required extension.

\section{Scope of the results}\label{sec:scope}
Using the projection theorem from~\cite{wan}, restated as Theorem~\ref{thm:gadget}, and the classical hardness of Gap Exact Set Cover, we have proved the block reduction, its two soundness bounds, the tensor amplification, the prime selection, and the deterministic norm transfer. The projection theorem is an established result used in the proof, rather than an additional assumption in the hardness statements.

The result is for every fixed approximation factor. The tensor order $T$, the source gap $\Gamma$, the gadget exponents, and their effective lower cutoffs may depend on that factor. If the desired factor grows with the input size, then $b,T$ also change, and the condition $\xi<1/(2T(b+1))$ forces a changing gadget exponent. This affects the base lattice size as well as the size increase from tensoring. The source reduction and the gadget construction would also require uniform estimates in their changing parameters. The present fixed-constant proof supplies no such uniform running-time bounds, so it does not establish growing-factor hardness under quasipolynomial- or subexponential-time reductions.

The finite-norm proof includes $p=1$. Its sign-family size depends on a finite moment order chosen from $p$, so the proof does not extend to $p=\infty$. Deterministic arbitrary-constant hardness for the infinite norm follows independently from Dinur's stronger theorem~\cite{dinur}.

Only Hamming distance of linear codes over a field is multiplied in the amplification proof. The integer tensor vectors themselves are controlled by Lemma~\ref{lem:seed} and Theorem~\ref{thm:phase}. No Lee-metric product formula or multiplicativity of lattice $\ell_1$ minima is required.

\Needspace{14\baselineskip}
\appendix
\section{Counterexample to the Lee-metric and lattice product formulas}\label{app:counterexample}
We give one example that disproves both the lattice $\ell_1$ identity and the Lee-metric quotient-code formula. In this appendix only, it is convenient to describe lattice vectors using a row basis:
\begin{equation}\label{eq:counter-basis}
 G=\begin{pmatrix}
 1&0&2&7\\
 0&1&7&5\\
 0&0&20&0\\
 0&0&0&20
 \end{pmatrix},\qquad\Lambda=\Z^4G.
\end{equation}
The column basis in the convention of the main text is $G^T$. Write $g_i$ for the rows of $G$.

\subsection{The lattice minimum is nine}
The vector
\[
 g_1+3g_2-g_3-g_4=(1,3,3,2)
\]
has $\ell_1$ norm $9$, so $\lambda(\Lambda)\le9$. Conversely, $(x,y,z,t)\in\Z^4$ belongs to $\Lambda$ exactly when
\[
 z\equiv2x+7y\pmod{20},\qquad
 t\equiv7x+5y\pmod{20}.
\]
Define the least absolute representative size by
\[
 \ell_{20}(a)=\min\{a,20-a\}\quad(0\le a<20),
\]
and extend it periodically to all integers. For residues $a,b$ of $x,y$, the smallest possible $\ell_1$ norm among representatives of the resulting four-coordinate residue is
\[
 f(a,b)=\ell_{20}(a)+\ell_{20}(b)
       +\ell_{20}(2a+7b)+\ell_{20}(7a+5b).
\]
The following exact finite table gives its minimum over $0\le b<20$, excluding $(a,b)=(0,0)$:
\begin{center}
\setlength{\tabcolsep}{5pt}
\begin{tabular}{c|rrrrrrrrrrr}
\toprule
$a$&0&1&2&3&4&5&6&7&8&9&10\\
\midrule
$\min_b f(a,b)$&9&9&9&9&9&9&12&10&13&14&18\\
\bottomrule
\end{tabular}
\end{center}
Each entry is the minimum of the twenty values of the displayed formula, or nineteen when $a=0$. The other rows follow from $f(-a,-b)=f(a,b)$. Thus every nonzero residue of a lattice vector modulo $20$ has Lee weight at least $9$. A nonzero integer vector with zero residue modulo $20$ has a coordinate of absolute value at least $20$. Both cases show that every nonzero vector of $\Lambda$ has $\ell_1$ norm at least $9$, and hence
\begin{equation}\label{eq:counter-minimum}
 \lambda(\Lambda)=9.
\end{equation}

\subsection{A tensor of weight eighty}
Consider the vector of $\Lambda\otimes\Lambda$ given by
\[
 w=g_1\otimes g_3+g_2\otimes g_4
    -g_3\otimes g_1-g_4\otimes g_2.
\]
Each $g_i$ is a $1\times4$ row vector, so $g_i\otimes g_j$ is a $1\times16$ row vector. Reshaping its consecutive blocks of four entries as rows gives the $4\times4$ outer-product matrix $g_i^Tg_j$, whose $(s,t)$ entry is $(g_i)_s(g_j)_t$. Thus $w$ reshapes as
\[
 W=g_1^Tg_3+g_2^Tg_4-g_3^Tg_1-g_4^Tg_2,
\]
where each product is ordinary matrix multiplication. Evaluating these products gives the nonzero matrix
\begin{equation}\label{eq:counter-tensor}
 W=\begin{pmatrix}
 0&0&20&0\\
 0&0&0&20\\
 -20&0&0&0\\
 0&-20&0&0
 \end{pmatrix}.
\end{equation}
There are four nonzero entries, each of absolute value $20$. Therefore
\begin{equation}\label{eq:counter-lattice}
 \lambda(\Lambda\otimes\Lambda)\le80<81=\lambda(\Lambda)^2.
\end{equation}
The universal lattice $\ell_1$ product identity is false. This example is consistent with Lemma~\ref{lem:seed}: its squared Euclidean norm is $1600$, which is greater than $\lambda(\Lambda)^2=81$.

\subsection{The Lee-metric quotient-code formula also fails}
For a vector $c$ modulo a positive integer $m$, its Lee weight $w_{\mathrm L}(c)$ is the sum of the least absolute values of its coordinate representatives. For a nonzero additive code $C$, let $d_{\mathrm L}(C)$ be the least Lee weight of a nonzero word.

In quotient-lattice notation, the Lee-metric product formula from Ojiro and Matsui~\cite{om} used in the first version would assert that, whenever $\Lambda_i$ have row Hermite-normal-form bases whose diagonal entries are strictly smaller than $m_i$ and $m_i\Z^{n_i}\subseteq\Lambda_i$,
\begin{equation}\label{eq:false-lee}
\begin{split}
 d_{\mathrm L}\bigl((\Lambda_1\otimes\Lambda_2)/
                    m_1m_2\Z^{n_1n_2}\bigr)
 \stackrel{?}{=}
 d_{\mathrm L}(\Lambda_1/m_1\Z^{n_1})\,
 d_{\mathrm L}(\Lambda_2/m_2\Z^{n_2}).
\end{split}
\end{equation}
Take $\Lambda_1=\Lambda_2=\Lambda$ from \eqref{eq:counter-basis}, and $m_1=m_2=40$. The displayed $G$ is already in row Hermite normal form: it is upper triangular, with positive pivots, and each entry above a pivot lies between zero and that pivot minus one. Its diagonal entries $1,1,20,20$ are all strictly below $40$.

Also, $20\Z^4\subseteq\Lambda$. For the first two standard basis vectors this follows from
\[
 20e_1=20g_1-2g_3-7g_4,\qquad
 20e_2=20g_2-7g_3-5g_4;
\]
the last two follow from $20e_3=g_3$ and $20e_4=g_4$. In particular, $40\Z^4\subseteq\Lambda$, as required.

Every balanced representative of a word of $\Lambda/40\Z^4$ still lies in $\Lambda$, because it differs from a lattice representative by an element of $40\Z^4$. A nonzero word has a nonzero balanced representative, so \eqref{eq:counter-minimum} bounds its Lee weight below by $9$. The vector $(1,3,3,2)$ remains a nonzero word of weight $9$ modulo $40$. Thus
\[
 d_{\mathrm L}(\Lambda/40\Z^4)=9.
\]
Tensoring the inclusions $40\Z^4\subseteq\Lambda$ gives
\[
 1600\Z^{16}\subseteq\Lambda\otimes\Lambda,
\]
so the tensor quotient in \eqref{eq:false-lee} is well-defined. The matrix $W$ in \eqref{eq:counter-tensor} remains nonzero modulo $1600$, and its entries $\pm20$ are already of least absolute value in their residue classes. Its Lee weight in this quotient is therefore $80$. Consequently,
\[
 d_{\mathrm L}\bigl((\Lambda\otimes\Lambda)/1600\Z^{16}\bigr)
 \le80<81
 =d_{\mathrm L}(\Lambda/40\Z^4)^2.
\]
All the stated modulus and diagonal conditions hold, yet the proposed Lee-metric product formula fails. The use of different moduli in the two quotients matters: the tensor displayed above would become zero modulo $20$, whereas it is a nonzero counterexample modulo $40^2$.

Finally, $\det G=400$. For every prime $q$ not dividing $400$, the matrix $G\bmod q$ is invertible, so $\Lambda\bmod q=\F_q^4$ has Hamming distance one. Thus choosing a large prime for this example does not create the modular-distance hypothesis used in Theorem~\ref{thm:phase}.

\Needspace{9\baselineskip}
\bigskip
\begin{flushleft}
\small
\textsc{Daqing Wan}\\
Center for Discrete Mathematics\\
College of Mathematics and Statistics\\
Chongqing University, Chongqing 401331, China\\
\textit{Email:} \href{mailto:dwan@math.uci.edu}{\nolinkurl{dwan@math.uci.edu}}
\end{flushleft}
\end{document}